\documentclass[12pt,reqno,fleqn]{amsart}
\usepackage{latexsym}
\usepackage{amssymb}
\usepackage{amsxtra}
\usepackage{amsmath}
\usepackage{amsfonts}
\newtheorem{theorem}{Theorem}[section]
\newtheorem{lemma}[theorem]{Lemma}
\newtheorem{proposition}[theorem]{Proposition}

\usepackage{amssymb}
\usepackage{mathrsfs}
\usepackage{amsmath}
\usepackage{enumerate}
\def\[{\begin{eqnarray}}
\def\]{\end{eqnarray}}

\renewcommand{\L}{\mathcal{L}}
\newcommand{\B}{\mathcal{B}}
\renewcommand{\d}{\partial}
\renewcommand{\S}{\mathcal{S}}
\newcommand{\M}{\mathcal{M}}
\newcommand{\Z}{\mathbb{Z}}
\newcommand{\N}{\mathbb{N}}
\newcommand\Zop{\mathbb{Z^{\mathrm{odd}}_+}}
\def\d{\partial}

\makeatother       

\makeatletter      
\@addtoreset{equation}{section}
\makeatother       

\begin{document}

\title[Quantum Torus symmetry]{Quantum Torus symmetry of the KP, KdV and BKP hierarchies}
\author{Chuanzhong Li\dag,\ Jingsong He\ddag }\dedicatory {
\mbox{}\hspace{-2.5cm}
\ \ \ \ \ Department of Mathematics, Ningbo University, Ningbo, 315211 Zhejiang, P.R.China\\
\dag lichuanzhong@nbu.edu.cn\\
\ddag hejingsong@nbu.edu.cn}
\date{}

\thanks{}

\begin{abstract}
In this paper, we construct the quantum Torus symmetry of the  KP hierarchy and further derive the quantum torus constraint on the tau function of the  KP hierarchy. That means we give a nice representation of the quantum Torus Lie algebra in the KP system by acting on its tau function.
Comparing to the $W_{\infty}$ symmetry, this quantum Torus symmetry has a nice algebraic structure with double indices.
 Further by reduction, we also construct the quantum Torus symmetries of the KdV and  BKP hierarchies and further derive the quantum Torus constraints on their tau functions.
These  quantum Torus constraints might have applications in the quantum field theory, supersymmetric gauge theory and so on.
\end{abstract}


\maketitle


\allowdisplaybreaks
\section{ Introduction}

The KP equation is an important shallow water wave equation which has a lot of application in plasma physics, water wave theory, topological field theory, string theory and so on. Later this equation was further generalized to a whole hierarchy which is called the KP hierarchy.
The KP hierarchy is one of the most important nonlinear integrable systems in mathematics and physics. It has nice structures such as the Virasoro type additional symmetry which was extensively studied in literature \cite{dickey2,os,heCKP,tiancCKP,maohuajmp}. Besides the Virasoro algebra,  another kind of infinite dimensional Lie algebras of Block type,
as generalizations of the well-known Virasoro algebra,  has
been studied intensively in \cite{Block,DZ,Su}. Recently as another generalization, the $W_{1+\infty}$ 3-algebra related to the KP hierarchy was constructed in \cite{zwz}.   After quantization, the q-discretization of the Virasoro algebra and the $q-W_{KP}^{(n)}$ algebra related to the KP hierarchy were considered in\cite{Qvirasoro,qwkp}.
In \cite{ourBlock}, we provide the Block type algebraic structure for the bigraded Toda hierarchy (BTH) \cite{ourJMP,solutionBTH}.
Later on, this Block type Lie algebra is found again in dispersionless bigraded Toda
 hierarchy \cite{dispBTH} and D type Drinfeld-Sokolov hierarchy\cite{dtyptds}. This Block algebra is sometimes called Torus Lie algebra which is a generalization of the Virasoro algebra.
 Similarly after quantization, the torus algebra becomes quantum torus algebra which can be realized in Fermion Fock space.

 The quantum torus algebra is identified with the sine-algebra \cite{sinealgebra} which is obtained
from $sl(N)$ by taking the large $N$ limit of its trigonometrical basis.
The quantum torus algebra is the Lie algebra derived
from a quantum non-commutative two-torus,  with
two generators $U,V$ satisfying the relation $UV = qVU$. Here the parameter $q$ is regarded as a
non-commutative parameter.  Let us normalize them as follows:
$v^{(k)}_m=q^{-\frac{km}2}U^mV^k.$
Then the normalized operators $v^{(k)}_m$ will satisfy the following structure of the quantum torus algebra without central extension,
\[[v^{(k)}_m,v^{(l)}_n]=(q^{\frac{lm-kn}2}-q^{-\frac{lm-kn}2})v^{(k+l)}_{m+n},\]
which is equivalent to
\[[\bar v^{(k)}_m,\bar v^{(l)}_n]=(q^{ml}-q^{nk})\bar v^{(k+l)}_{m+n},\]
with $\bar v^{(k)}_m=q^{-\frac{mk}2}v^{(k)}_m=q^{-km}U^mV^k$.

Then after the second-quantizations of $v^{(k)}_m$
by means of the
fermions generating function $\phi(z),\phi(z)^*$ as
\[V^{(k)}_m
= 	\oint \frac {dz}{2\pi i}:\phi(z)v^{(k)}_m\phi(z)^*:,\]
the operators $V^{(k)}_m$ will satisfy the quantum torus algebra with central extension as
\[[V^{(k)}_m,V^{(l)}_n]=(q^{\frac{lm-kn}2}-q^{-\frac{lm-kn}2})(V^{(k+l)}_{m+n}-\delta_{m+n,0}\frac{q^{k+l}}{1-q^{k+l}}).\]

The positive half of quantum torus Lie algebra, i.e. $k\geq 0$ is a quantum cylinder Lie algebra
 obtained from a quantum cylinder which becomes a classical cylinder $C^*$
in the context of the random plane partition. Also the double covers of the cylinder have important application in the Seiberg-Witten hyper-elliptic curves of five-dimensional $N = 1$
supersymmetric gauge theories \cite{siebergwitten,instanton} at the thermodynamic limit or the
semi-classical limit. In \cite{takasakiquantum},  it
reveals a remarkable connection among random plane partition, melting Crystal, quantum torus
Lie algebra, and Toda hierarchy.

The KdV hierarchy is an important reduced hierarchy from the KP hierarchy which is also an important member of the Gelfand-Dichey hierarchy. It is proved to have a nice Virasoro symmetry \cite{dickeyacta}. The Virasoro constraint on the KdV hierarchy has important application in topological field theory and Gromov-Witten invariant theory \cite{witten,kontsevich,Douglas}.
Also as an important
sub-hierarchy of the KP hierarchy, BKP hierarchy
\cite{DKJM-KPBKP,kt1}, has been shown to possess
additional symmetries with consideration on
the reductions on the Lax operator \cite{tu}.

Our main purpose of this article is to give the complete quantum torus
flows on the Lax operator, wave function and tau function of the  KP hierarchy which form the positive half of quantum torus algebra.
 In \cite{kodamaWalgebra}, the additional symmetries of the KP hierarchy were generalized
to a $W_{1+\infty}$ algebra with a
complicated algebraic relation. The quantum torus algebra under consideration in this paper is very simple and elegant. By acting on the tau function, we give a nice representation of the quantum torus Lie algebra which is an uneasy work in the field of the representation theory of the Lie algebra. Further we will try to generalize the results to the KdV hierarchy and the BKP hierarchy.

This paper is organized as follows. We give a brief description of the
 KP hierarchy in Section 2. The main results are presented and proved in Section 3,
which concerns the quantum torus symmetry of the  KP hierarchy. Basing on the Section 3, the quantum torus constraint on the tau function space will be given in  Section 4. The quantum torus symmetry is further generalized to the KdV hierarchy and BKP hierarchy in Section 5 and Section 6. The Section 7 will be devoted
to conclusions and discussions.



\section{ The KP hierarchy}
To start the principle content of this paper, we need to recall some basic knowledge related to the KP hierarchy \cite{dickey2,DKJM-KPBKP}.
At the beginning  of the recalling,  one need to know the following fundamental
Leibnitz rule over pseudo-differential operators' space
\begin{equation}
     \partial^n \circ f=\sum_{k\ge0}\binom{n}{k}(\partial^kf)\partial^{n-k},\qquad n\in \Z.
     \end{equation}
\noindent For any pseudo-differential operator $W=\sum\limits_ip_i(x)\partial^i$, the conjugate operation ``$*$''
for $W$ is defined by $W^*=\sum\limits_i(\partial^*)^ip_i(x)$ with
\[&&\partial^*=-\partial,\ \ (\partial^{-1})^*=(\partial^*)^{-1}=-
\partial^{-1}.\]
 The Lax operator $L$ of the KP hierarchy is given by
\begin{equation}\label{qkplaxoperator}
L=\partial+
u_{-1}\partial^{-1}+u_{-2}\partial^{-2}+\cdots.
\end{equation}
where $u_i=u_i(x,t_1, t_2, t_3,\cdots,),i=-1,-2, -3, \cdots $. The
corresponding Lax equation of the KP hierarchy is defined as
\begin{equation}
\dfrac{\partial L}{\partial t_n}=[B_n, L], \ \ n=1, 2, 3, \cdots,
\end{equation}
where the differential projective operator $B_n=(L^n)_+=\sum\limits_{i=0}^n
b_i\partial^i$.  The Lax operator $L$ in
eq.(\ref{qkplaxoperator}) can be generated by the dressing operator
$S=1+ \sum_{k=1}^{\infty}\omega_k
\partial^{-k}$ in the following way
\begin{equation}\label{dressingst}
L=S  \partial  S^{-1}.
\end{equation}
The dressing operator $S$ satisfies the following Sato equation
\begin{equation}\label{satoequation}
\dfrac{\partial S}{\partial t_n}=-(L^n)_-S, \quad n=1,2, 3, \cdots.
\end{equation}

By using the above dressing structure \eqref{dressingst} and the Sato equation \eqref{satoequation}, it is convenient to construct the Orlov-Shulman's  operator which is used to give the additional symmetry of the KP hierarchy in the next section.

\section{Quantum Torus symmetry of the  KP hierarchy}
In this section, we shall construct the additional symmetry of the  KP hierarchy and discuss the algebraic structure of the
additional symmetry  flows.

To this end, firstly we define the following operator $\Gamma$ and the Orlov-Shulman's  operator $M$ like in \cite{os} as
 \begin{equation}
\Gamma=\sum_{i=1}^{\infty}it_i\partial^{i-1},\ \ M= S \Gamma S^{-1}.
\end{equation}
The Lax operator $L$ and the Orlov-Shulman's  operator $M$ satisfy the following canonical relation
\[[L,M]=1.\]
Then basing on a quantum parameter $q$, the additional flows for the time variable $t_{m,n},t_{m,n}^*$ are
defined respectively as follows
\begin{equation}
\dfrac{\partial S}{\partial t_{m,n}}=-(M^mL^n)_-S,\ \dfrac{\partial S}{\partial t^*_{m,n}}=-(e^{mM}q^{nL})_-S,\ m,n \in \N,
\end{equation}
or equivalently rewritten as
\begin{equation}
\dfrac{\partial L}{\partial t_{m,n}}=-[(M^mL^n)_-,L], \qquad
\dfrac{\partial M}{\partial t_{m,n}}=-[(M^mL^n)_-,M],
\end{equation}

\begin{equation}
\dfrac{\partial L}{\partial t^*_{m,n}}=-[(e^{mM}q^{nL})_-,L], \qquad
\dfrac{\partial M}{\partial t^*_{m,n}}=-[(e^{mM}q^{nL})_-,M].
\end{equation}

 Generally, one can also derive
\begin{equation}\label{MLK}
\partial_{t^*_{l,k}}(e^{mM}q^{nL})=[-(e^{lM}q^{kL})_-,e^{mM}q^{nL}].
\end{equation}

One can find the operators' set $\{e^{mM}q^{nL},\ m,n\geq 0\}$ has an isomorphism with the operators' set $\{q^{nz}e^{m\partial_z},\ m,n\geq 0\}$ as
\begin{equation}
e^{mM}q^{nL} \qquad \mapsto\qquad  q^{nz}e^{m\partial_z},
\end{equation}
with the following commutator
\begin{equation}
[q^{nz}e^{m\partial_z},q^{lz}e^{k\partial_z}]=(q^{ml}-q^{nk})q^{(n+l)z}e^{(m+k)\partial_z}.
\end{equation}

It is known \cite{dickeyacta} that  the additional flows $\dfrac{\partial }{\partial
t_{m,n}}$  commute with the flows $\dfrac{\partial
}{\partial t_k}$, i.e. $[\dfrac{\partial }{\partial
t_{m,n}},\dfrac{\partial
}{\partial t_k}]=0$,  but do not
commute with each other, and form a kind of $W_{\infty}$ infinite dimensional  additional Lie  symmetries.
 This further leads to the commutativity of the additional flows $\dfrac{\partial }{\partial
t_{m,n}^*}$  with the flows $\dfrac{\partial
}{\partial t_k}$ and the additional flows $\dfrac{\partial }{\partial
t_{m,n}^*}$  themselves constitute the quantum torus algebra.

Next we shall prove the commutativity between the additional
flows $\partial_{t_{l,k}}(i.e. \dfrac{\partial }{\partial
t_{l,k}})$  and the original flows $\partial_{t_n}(i.e. \dfrac{\partial }{\partial t_n})$ of the KP hierarchy.
\begin{proposition}
The additional flows of $\partial_{t_{l,k}}$ are  symmetry flows of the  KP hierarchy, i.e. they commute with all $\partial_{t_n}$ flows of the   KP hierarchy \cite{dickey2,dickeyacta}.
\end{proposition}
\begin{proof}
According the action of  $\partial_{t_{l,k}}$ and $\partial_{t_n}$ on the
dressing operator $S$, then
\begin{eqnarray*}
[\partial_{t_{l,k}},\partial_{t_n}]S &=& -\partial_{t_{l,k}}(L_-^n
S)-\partial_{t_n}[-(M^l L^k)_-]S \\
&=&(-\partial_{t_{l,k}}L_-^n)S-L^n_-\partial_{t_{l,k}}S-[(M^l
L^k)_-]L_-^nS \\
 &&+[L^n_+,M^l L^k]_-S\\
 &=&0.
\end{eqnarray*}
Therefore the proposition holds.
\end{proof}
With the help of this proposition, we can derive the following theorem.
\begin{theorem}
The additional flows $\partial_{t^*_{l,k}}$ are  symmetries of the  KP hierarchy, i.e. they commute with all $\partial_{t_n}$ flows of the   KP hierarchy.
\end{theorem}
\begin{proof}
According to the action of  $\partial_{t^*_{l,k}}$ and $\partial_{t_n}$ on the
dressing operator $S$,  we can rewrite the quantum torus flow $\partial_{t^*_{l,k}}$¡¡in terms of a combination of $\partial_{t_{p,s}}$ flows
\begin{eqnarray*}
\partial_{t^*_{l,k}}S &=& -(\sum_{p,s=0}^{\infty}\frac{l^p(k\log q)^sM^pL^s}{p!s!})_-S\\
 &=&\sum_{p,s=0}^{\infty}\frac{l^p(k\log q)^s}{p!s!}\partial_{t_{p,s}}S,
\end{eqnarray*}
which further leads to
\begin{eqnarray*}
[\partial_{t^*_{l,k}},\partial_{t_n}]S &=& [\sum_{p,s=0}^{\infty}\frac{l^p(k\log q)^s}{p!s!}\partial_{t_{p,s}},\partial_{t_n}]S\\
&=& \sum_{p,s=0}^{\infty}\frac{l^p(k\log q)^s}{p!s!}[\partial_{t_{p,s}},\partial_{t_n}]S\\
 &=&0.
\end{eqnarray*}
Therefore the theorem holds.
\end{proof}
Because
\begin{eqnarray*}
[z^s\partial^p,z^b\partial^a]=\sum_{\alpha\beta}C_{\alpha\beta}^{(ps)(ab)}z^{\beta}\partial^{\alpha},
\end{eqnarray*}
and
\begin{equation}\label{zformala}
[q^{nz}e^{m\partial_z},q^{lz}e^{k\partial_z}]=(q^{ml}-q^{nk})q^{(n+l)z}e^{(m+k)\partial_z},
\end{equation}
therefore comparing the terms  with $z^{\alpha}\partial^{\beta}$ on both sides of eq.\eqref{zformala} will lead to the following identity
\begin{eqnarray}\label{combina}
&&\sum_{p,s=0}^{\infty}\sum_{a,b=0}^{\infty}\frac{n^p(m\log q)^s}{p!s!}\frac{l^a(k\log q)^b}{a!b!}C_{\alpha\beta}^{(ps)(ab)}=(q^{ml}-q^{nk})\frac{(n+l)^\alpha((m+k)\log q)^\beta}{\alpha!\beta!}.
\end{eqnarray}

Now it is time to identify the algebraic structure of the quantum torus
additional $\partial_{t_{l,k}^*}$ flows of the  KP hierarchy in the following theorem.
\begin{theorem}
The additional flows $\partial_{t_{l,k}^*}$ of the  KP hierarchy form the
positive half of the quantum torus algebra, i.e.,
\begin{equation}
[\partial_{t^*_{n,m}},\partial_{t^*_{l,k}}]=(q^{ml}-q^{nk})\partial_{t^*_{n+l,m+k}},\ \ n,m,l,k\geq 0.
\end{equation}

\end{theorem}
\begin{proof}
Using the Jacobi identity, one can derive the following computation which will finish the proof of this theorem

\begin{eqnarray*}
&&[\partial_{t^*_{n,m}},\partial_{t^*_{l,k}}]L\\
&=&\partial_{t^*_{n,m}}([-(e^{lM}q^{kL})_-,L])-\partial_{t^*_{l,k}}([-(e^{nM}q^{mL})_-,L]) \\
&=&[-(\partial_{t^*_{n,m}} (e^{lM}q^{kL}))_-,L] +[-(e^{lM}q^{kL})_-,(\partial_{t^*_{n,m}} L)]
+[[-(e^{lM}q^{kL})_-,e^{nM}q^{mL}]_-,L]\\
&&+[(e^{nM}q^{mL})_-,[-(e^{lM}q^{kL})_-,L]] \\
&=&[[(e^{nM}q^{mL})_-,e^{lM}q^{kL}]_-,L] +[(e^{lM}q^{kL})_-,[(e^{nM}q^{mL})_-,L]]
+[[-(e^{lM}q^{kL})_-,e^{nM}q^{mL}]_-,L]\\
&&+[(e^{nM}q^{mL})_-,[-(e^{lM}q^{kL})_-,L]] \\
&=&[[(e^{nM}q^{mL})_-,e^{lM}q^{kL}]_-,L] +[[(e^{lM}q^{kL})_-,(e^{nM}q^{mL})_-],L]
+[[-(e^{lM}q^{kL})_-,e^{nM}q^{mL}]_-,L]\\
&=&[[e^{nM}q^{mL},e^{lM}q^{kL}]_-,L]\\
&=&-(q^{ml}-q^{nk})[(e^{(n+l)M}q^{(m+k)L})_-,L]\\
&=&(q^{ml}-q^{nk})\partial_{t^*_{n+l,m+k}}L.
\end{eqnarray*}

One can also prove this theorem as following in another way by rewriting the quantum torus flows in terms of a combination of $t_{m,n}$ flows
\begin{eqnarray*}
&&[\partial_{t^*_{n,m}},\partial_{t^*_{l,k}}]L\\
&=&[\sum_{p,s=0}^{\infty}\frac{n^p(m\log q)^s}{p!s!}\partial_{t_{p,s}},\sum_{a,b=0}^{\infty}\frac{l^a(k\log q)^b}{a!b!}\partial_{t_{a,b}}]L\\
&=&\sum_{p,s=0}^{\infty}\sum_{a,b=0}^{\infty}\frac{n^p(m\log q)^s}{p!s!}\frac{l^a(k\log q)^b}{a!b!}[\partial_{t_{p,s}},\partial_{t_{a,b}}]L\\
&=&\sum_{p,s=0}^{\infty}\sum_{a,b=0}^{\infty}\frac{n^p(m\log q)^s}{p!s!}\frac{l^a(k\log q)^b}{a!b!}\sum_{\alpha\beta}C_{\alpha\beta}^{(ps)(ab)}\partial_{t_{\alpha,\beta}}L\\
&=&(q^{ml}-q^{nk})\sum_{\alpha,\beta=0}^{\infty}\frac{(n+l)^\alpha((m+k)\log q)^\beta}{\alpha!\beta!}\partial_{t_{\alpha,\beta}}L\\
&=&(q^{ml}-q^{nk})\partial_{t^*_{n+l,m+k}}L.
\end{eqnarray*}
\end{proof}
Till now, we can find the $\partial_{t^*_{l,k}}$ additional flows constitute a nice quantum torus algebra. A natural question is whether we can get the  quantum torus constraint which is a generalization of the well-known Virasoro constraint. The answer will be given in the next section.

\section{Quantum Torus constraints on the tau function}
Acting on the wave function $\phi$ of the KP hierarchy, one can rewrite the Lax equation of the KP hierarchy using linear equations
\[L\phi=\lambda\phi,\ \ \frac{\d \phi}{\d t_n}=L^n_+\phi.\]

Then the tau function of the KP hierarchy can be defined as \cite{DKJM-KPBKP}
\[\phi=\frac{e^{\eta}\tau}{\tau}e^{\sum_{k=1}t_k\lambda^k},\]
where \[\eta=\sum_{i=1}^{\infty}\frac{\lambda^{-i}}{i}\frac{\partial}{\partial t_i}.\]

Alder, Shiota and van Moerbeke \cite{asv94,asv2} have shown that
\begin{eqnarray}
\partial_{t_{p,s}}\log \phi=(e^{\eta}-1)\frac{\frac{W_s^{(p+1)}}{p+1}(\tau)
}{\tau},
\end{eqnarray}
where $W_s^{(p+1)}$ is the generator of $W_{\infty}$ algebra.
Then by using
\begin{eqnarray*}
\partial_{t^*_{l,k}}\log \phi&=&\sum_{p,s=0}^{\infty}\frac{l^p(k\log q)^s}{p!s!}\partial_{t_{p,s}}\log \phi.
\end{eqnarray*}
and defining \[L_{l,k}:=\sum_{p,s=0}^{\infty}\frac{l^p(k\log q)^s}{p!s!}\frac{W_s^{(p+1)}}{p+1},\]
we get
\begin{eqnarray}
\partial_{t^*_{l,k}}\log \phi=(e^{\eta}-1)\frac{L_{l,k}(\tau)
}{\tau}.
\end{eqnarray}
The quantum torus constraint on the wave function $\phi$, i.e.
\begin{eqnarray}
\partial_{t^*_{l,k}}\phi&=&0,
\end{eqnarray}
will lead to the quantum torus constraint on the tau function
\[L_{l,k}\tau=c,\] \label{constraintontauKP}
where $c$ is a constant.

Basing on the commuting relation among operators $W_s^{(p+1)}$ and the formula eq.\eqref{combina}, one can prove the operators $\{L_{l,k}, l,k\geq 0\}$ constitute a quantum torus algebra by acting on tau function space, i.e.
\begin{equation}
[L_{n,m},L_{l,k}]=(q^{ml}-q^{nk})L_{n+l,m+k},\ \ n,m,l,k\geq 0.
\end{equation}
Till now, we announce that a representation of the quantum torus algebra was found, i.e. $\{L_{l,k},\ l,k\geq 0\}.$
Also one can find the quantum torus operator $L_{l,k}$ has an infinite number of terms. What is the application of the quantum torus constraint in the Seiberg-Witten theory, supersymmetric gauge theory and so on might be an interesting question.

\section{ The KdV hierarchy and its quantum torus symmetry }

Similar to the general way in describing the classical KdV hierarchy
\cite{dickey2}, we will give a brief introduction of the KdV hierarchy as a reduction of the KP hierarchy.

The Lax operator $\L$ of the KdV hierarchy is given by
\begin{equation}\label{kdvqkplaxoperator}
\L=L^2=\partial^2+ u.
\end{equation}
 The
corresponding Lax equations of the KdV hierarchy are defined as
\begin{equation}
\dfrac{\partial \L}{\partial t_n}=[\B_n, \L], \ \ n=1, 3, 5, \cdots,
\end{equation}
where the differential part $\B_n=(\L^{\frac n2})_+=\sum\limits_{i=0}^n
c_i\partial^i$.  The Lax operator $\L$ in
eq.(\ref{kdvqkplaxoperator}) can be generated by a dressing operator
$\S=1+ \sum_{k=1}^{\infty}\tilde\omega_k
\partial^{-k}$ in the following way
\begin{equation}
\L=\S  \partial^2  \S^{-1}.
\end{equation}
The dressing operator $\S$ satisfies the following Sato equation
\begin{equation}
\dfrac{\partial \S}{\partial t_n}=-(\L^{\frac n2})_-\S, \quad n=1,3,5, \cdots.
\end{equation}

After the above preparation, in the next part, we shall construct the additional symmetry and discuss the algebraic structure of the
additional symmetry  flows of the  KdV hierarchy.

To this end, firstly one define $\Gamma_{kdv}$ and the Orlov-Shulman's  operator $\M$ as \cite{dickeyacta}
 \begin{equation}
\Gamma_{kdv}=\frac12\sum_{i\in \Zop }it_i\partial^{i-2},\ \ \M= \S \Gamma_{kdv} \S^{-1}.
\end{equation}
The Lax operator $\L$ and the Orlov-Shulman's $\M$ operator satisfy the following canonical relation
\[[\L,\M]=1.\]
Then basing on a quantum parameter $q$, the additional flows for the time variable $t_{m,n}^*$ are
defined as follows
\begin{equation}
 \dfrac{\partial \S}{\partial t^*_{m,n}}=-(e^{m\M}q^{n\L})_-\S,\ m,n \in \Z_+,
\end{equation}
or equivalently rewritten as

\begin{equation}
\dfrac{\partial \L}{\partial t^*_{m,n}}=-[(e^{m\M}q^{n\L})_-,\L], \qquad
\dfrac{\partial \M}{\partial t^*_{m,n}}=-[(e^{m\M}q^{n\L})_-,\M].
\end{equation}

 Generally, one can also derive
\begin{equation}\label{kdvMLK}
\partial_{t^*_{l,k}}(e^{m\M}q^{n\L})=[-(e^{l\M}q^{k\L})_-,e^{m\M}q^{n\L}].
\end{equation}

Similarly as the KP hierarchy, one can further derive the commutativity of the additional flows $\dfrac{\partial }{\partial
t_{m,n}^*}$  with the flows $\dfrac{\partial
}{\partial t_k}$ and the additional flows $\dfrac{\partial }{\partial
t_{m,n}^*}$  themselves form quantum torus algebra which is included in the following theorem.

\begin{theorem}
The additional flows of $\partial_{t^*_{l,k}}$ are  symmetries of the  KdV hierarchy, i.e. they commute with all $\partial_{t_n}$ flows of the   KdV hierarchy.
\end{theorem}
\begin{proof}
The proof is similar as the KP hierarchy which will be omitted here.
\end{proof}

Now it is time to identity the algebraic structure of the quantum torus
additional symmetry $t_{l,k}^*$ flows of the  KdV hierarchy which is similar as the KP hierarchy.
\begin{theorem}\label{kdvalg}
The additional flows $\partial_{t_{l,k}^*}$ of the  KdV hierarchy constitute the
 quantum torus algebra, i.e.,
\begin{equation}
[\partial_{t^*_{n,m}},\partial_{t^*_{l,k}}]=(q^{ml}-q^{nk})\partial_{t^*_{n+l,m+k}},\ \ n,m,l,k\geq 0.
\end{equation}

\end{theorem}
\begin{proof}
The proof is similar as the KP hierarchy which will be omitted here.
\end{proof}
Till now, we can find the $t^*_{l,k}$ additional flows constitute a nice quantum torus algebra. Acting on the wave function $\psi$, one can rewrite the Lax equation of the KdV hierarchy using linear equations
\[\L\psi=\lambda\psi,\ \ \frac{\d \psi}{\d t_{2n+1}}=\L^{\frac{2n+1}2}_+\psi.\]

Then the tau function of the KdV hierarchy can be defined as \cite{dickeyacta}
\[\psi=\frac{e^{\tilde\eta}\tau_{kdv}}{\tau_{kdv}}e^{\sum_{k=1}t_{2k-1}\lambda^{2k-1}},\]
where \[\tilde\eta=\sum_{i\in \Zop}\frac{\lambda^{-i}}{i}\frac{\partial}{\partial t_i}.\]

As we all know, the action on the wave function $\psi$ can be expressed on the tau function space as following
\begin{eqnarray}
\partial_{t_{p,s}}\log \psi=(e^{\tilde\eta}-1)\frac{\frac{W_{kdv,s}^{(p+1)}}{p+1}(\tau_{kdv})
}{\tau_{kdv}}.
\end{eqnarray}
where $W_{kdv,s}^{(p+1)}$ is the generator of $W_{\infty}$ algebra for the KdV hierarchy \cite{dickeyacta}.
Then denote \[\tilde L_{l,k}:=\sum_{p,s=0}^{\infty}\frac{l^p(k\log q)^s}{p!s!}\frac{W_{kdv,s}^{(p+1)}}{p+1},\]
and we can get the quantum torus constraint on the tau function $\tau_{kdv}$ of the KdV hierarchy
\[\tilde L_{l,k}\tau_{kdv}=c,\]
where $c$ is a constant.

\section{ The BKP hierarchy and its quantum torus constraint}

Similar to the general way in describing the classical the BKP hierarchy
\cite{DKJM-KPBKP,dickey2}, we will give a brief introduction of the BKP hierarchy.

Basing on the definition, the Lax operator of the BKP hierarchy has  form
\begin{equation} \label{PhP}
\L_B= \d+\sum_{i\ge1}v_i  \d^{-i},
\end{equation}
 such that

\begin{equation}\label{Bcondition}
\L_B^*=-\d \L_B\d^{-1}.
\end{equation}
We call eq.\eqref{Bcondition} the B type condition of the BKP hierarchy.

The  BKP hierarchy is defined by the following
Lax equations:
\begin{align}\label{bkpLax}
& \frac{\d  \L_B}{\d t_k}=[(\L_B^k)_+,  \L_B],   \ k\in\Zop.
\end{align}

Note that $\d/\d t_1$ flow is equivalent to $\d/\d x$ flow, therefore it is reasonable to
assume $t_1=x$ in the next sections. The operator $\L_B$ can be generated by a dressing operator
$\Phi_B=1+ \sum_{k=1}^{\infty}\bar\omega_k
\partial^{-k}$ in the following way
\begin{equation}
\L_B=\Phi_B  \partial \Phi_B^{-1},
\end{equation}
where $\Phi_B$
 satisfies
\begin{equation}\label{phipsi}
\Phi_B^*= \d\Phi_B^{-1} \d^{-1}.
\end{equation}
The dressing operator $\Phi_B$ needs to satisfy the following Sato equations
\begin{equation}
\dfrac{\partial \Phi_B}{\partial t_n}=-(\L_B^{n})_-\Phi_B, \quad n=1,3,5, \cdots.
\end{equation}

Using the above dressing structure and Sato equations, it is convenient to construct the Orlov-Shulman's  operator which is used to give the quantum torus type additional symmetry of the  BKP hierarchy.
In the next part, we shall aim at constructing the additional symmetry and discuss the algebraic structure of the
additional  flows of the  BKP hierarchy.

To this end, firstly we define the operator $\Gamma_B$ and the Orlov-Shulman's  operator $\M_B$ like in \cite{os} as
 \begin{equation}
\Gamma_B=\sum_{i\in \Zop }it_i\partial^{i-1},\ \ \M_B= \Phi_B \Gamma_B \Phi_B^{-1}.
\end{equation}
The Lax operator $\L_B$ and the Orlov-Shulman's $\M_B$ operator satisfy the following canonical relation
\[[\L_B,\M_B]=1.\]

Given an operator $\L_B$, the dressing operators $\Phi_B$ are determined uniquely up to a multiplication to the
right by operators with
constant coefficients.

We denote $t=(t_1,t_3,t_5,\dots)$ and introduce
the wave function as
\begin{align}\label{wavef}
w_B(t; z)=\Phi_B e^{\xi_B(t;z)},
\end{align}
where the function $\xi_B$ is defined as $\xi_B(t;
z)=\sum_{k\in\Zop} t_k z^k$. It is easy to see
$\d^i e^{x z}=z^i e^{x z},\ \ i\in\Z$
and
\[
\L_B\,w_B(t;z)=z w_B(t;z),\ \ \frac{\d w_B}{\d t_{2n+1}}=\L^{2n+1}_{B+}w_B.
\]

The tau function of the   BKP hierarchy can be defined in form of the wave functions as
\begin{align}\label{wtau}
w_B(t,z)=\frac{\tau_B(t-2[z^{-1}])}{\tau_B(t)}
e^{\xi_B(t;z)},
\end{align}
where $[z]=\left(z,z^3/3,z^5/5,\dots\right)$.

With the above preparation, it is time to  construct additional symmetries for the  BKP hierarchy in the next part.
Then it is easy to get  that the operator $\M_B$  satisfy

\begin{equation}
[\L_B, \M_B]=1,  \
\M_B w_B(z)=\d_z w_B(z);
\end{equation}
\begin{equation}\label{bkpMt}
\frac{\d \M_B}{\d t_k}=[(\L_B^k)_+,\M_B],\ \ k\in\Zop.
\end{equation}

Given any pair of integers $(m,n)$ with $m,n\ge0$, we will introduce the following operator $B_{m n}$
\begin{align}\label{defBoperator}
B_{m n}=\M_B^m\L_B^{n}-(-1)^{n} \L_B^{n-1}\M_B^{m}\L_B.
\end{align}

For any operator $B_{m n}$ in \eqref{defBoperator}, one has
\begin{align}\label{Bflow}
&\frac{\d B_{m n}}{\d t_k}=[(\L_B^k)_+, B_{m n}],  \ k\in\Zop.
\end{align}

To prove that $B_{m n}$  satisfy B type condition, we need the following lemma.
\begin{lemma}\label{BtypM}
The operator $\M_B$ satisfies the following identity,
\[\label{MBproperty}
\M_B^*
=\d\L_B^{-1}\M_B\L _B\d^{-1}.\]
\end{lemma}
\begin{proof}
Using
\[
 \Phi_B^*=\d\Phi_B^{-1} \d^{-1},\ \ \Gamma_B^*=\Gamma_B;\]
the  following calculations
\[
\M_B^* =\Phi_B^{*-1}\Gamma_B \Phi_B^*=\d\Phi_B \d^{-1}\Gamma_B \d\Phi_B^{-1} \d^{-1}
=\d\Phi_B \d^{-1}\Phi_B^{-1}\M_B\Phi_B \d\Phi_B^{-1} \d^{-1},\]
will lead to \eqref{MBproperty}.
\end{proof}
 Basing on the Lemma \ref{BtypM} above, it is easy to check that the operator
$B_{m n}$  satisfy the B type condition, namely
\begin{equation}\label{btypeB}
B_{m n}^*=-\d  B_{m n} \d^{-1}.
\end{equation}

Now we will denote the operator $D_{m n}$ as
\begin{equation}
D_{m n}:=e^{m\M_B}q^{n\L_B}-\L_B^{-1}q^{-n\L_B}e^{m\M_B}\L_B,
\end{equation}
which further leads to
\begin{equation}
D_{m n}=\sum_{p,s=0}^{\infty}\frac{m^p(n\log q)^s(\M_B^p\L_B^s-(-1)^s\L_B^{s-1}\M_B^p\L_B)}{p!s!}=\sum_{p,s=0}^{\infty}\frac{m^p(n\log q)^sB_{p s}}{p!s!}.
\end{equation}
Using eq. \eqref{btypeB}, the following calculation will lead to the B type anti-symmetry property of $D_{m n}$ as
\begin{eqnarray*}D_{m n}^*
&=&(\sum_{p,s=0}^{\infty}\frac{m^p(n\log q)^sB_{p s}}{p!s!})^*\\
&=&-(\sum_{p,s=0}^{\infty}\frac{m^p(n\log q)^s\d B_{p s}\d^{-1}}{p!s!})\\
&=&-\d(\sum_{p,s=0}^{\infty}\frac{m^p(n\log q)^sB_{p s}}{p!s!})\d^{-1}\\
&=&-\d  D_{m n} \d^{-1}.
\end{eqnarray*}
Therefore we get the following important B type condition which the operator $D_{m n}$ satisfies
\begin{equation}
D_{m n}^*=-\d  D_{m n} \d^{-1}.
\end{equation}

Then basing on a quantum parameter $q$, the additional flows for the time variable $t_{m,n},t_{m,n}^*$ are
defined as follows
\begin{equation}
\dfrac{\partial \Phi_B}{\partial t_{m,n}}=-(B_{m n})_-\Phi_B,\ \
 \dfrac{\partial \Phi_B}{\partial t^*_{m,n}}=-(D_{m n})_-\Phi_B.
\end{equation}
or equivalently rewritten as

\begin{equation}
\dfrac{\partial \L_B}{\partial t_{m,n}}=-[(B_{m n})_-,\L_B], \qquad
\dfrac{\partial \M_B}{\partial t^*_{m,n}}=-[(D_{m n})_-,\M_B].
\end{equation}

 Generally, one can also derive
\begin{equation}\label{bkpMLK}
\partial_{t^*_{l,k}}(D_{m n})=[-(D_{l k})_-,D_{m n}].
\end{equation}
\noindent {\bf Remark:}
The specific construction of the operator $D_{m n}$ shows the impact of the reduction condition in eq. (\ref{Bcondition})  on the generators of the additional flows.

 This further leads to the commutativity of the additional flow $\dfrac{\partial }{\partial
t_{m,n}^*}$  with the flow $\dfrac{\partial
}{\partial t_k}$ in the following theorem.

\begin{theorem}
The additional flows of $\partial_{t^*_{l,k}}$ are  symmetries of the  BKP hierarchy, i.e. they commute with all $\partial_{t_n}$ flows of the   BKP hierarchy.
\end{theorem}
\begin{proof}
The proof is similar as the KP hierarchy by using the Proposition 3 in \cite{tu}, i.e.
the additional flows of $\partial_{t_{l,k}}$ can commute with all $\partial_{t_n}$ flows of the   BKP hierarchy. The detail will be omitted here.
\end{proof}
The additional flows $\partial_{t_{l,k}}$ of the  BKP hierarchy form the
$W_{\infty}$ algebra \cite{tu}
\begin{eqnarray*}
&&[\partial_{t_{p,s}},\partial_{t_{a,b}}]\L_B=\sum_{\alpha\beta}C_{\alpha\beta}^{(ps)(ab)}\partial_{t_{\alpha,\beta}}\L_B.
\end{eqnarray*}

Now it is time to identity the algebraic structure of the
additional $t_{l,k}^*$ flows of the  BKP hierarchy.
\begin{theorem}\label{bkpalg}
The additional flows $\partial_{t_{l,k}^*}$ of the  BKP hierarchy form the
positive half of quantum torus algebra, i.e.,
\begin{equation}
[\partial_{t^*_{n,m}},\partial_{t^*_{l,k}}]=(q^{ml}-q^{nk})\partial_{t^*_{n+l,m+k}},\ \ n,m,l,k\geq 0.
\end{equation}

\end{theorem}

\begin{proof}

One can also prove this theorem as following by rewriting the quantum torus flow in terms of a combination of $\partial_{t_{m,n}}$ flows
\begin{eqnarray*}
&&[\partial_{t^*_{n,m}},\partial_{t^*_{l,k}}]\L_B\\
&=&[\sum_{p,s=0}^{\infty}\frac{n^p(m\log q)^s}{p!s!}\partial_{t_{p,s}},\sum_{a,b=0}^{\infty}\frac{l^a(k\log q)^b}{a!b!}\partial_{t_{a,b}}]\L_B\\
&=&\sum_{p,s=0}^{\infty}\sum_{a,b=0}^{\infty}\frac{n^p(m\log q)^s}{p!s!}\frac{l^a(k\log q)^b}{a!b!}[\partial_{t_{p,s}},\partial_{t_{a,b}}]\L_B\\
&=&\sum_{p,s=0}^{\infty}\sum_{a,b=0}^{\infty}\frac{n^p(m\log q)^s}{p!s!}\frac{l^a(k\log q)^b}{a!b!}\sum_{\alpha\beta}C_{\alpha\beta}^{(ps)(ab)}\partial_{t_{\alpha,\beta}}\L_B\\
&=&(q^{ml}-q^{nk})\sum_{\alpha,\beta=0}^{\infty}\frac{(n+l)^\alpha((m+k)\log q)^\beta}{\alpha!\beta!}\partial_{t_{\alpha,\beta}}\L_B\\
&=&(q^{ml}-q^{nk})\partial_{t^*_{n+l,m+k}}\L_B.
\end{eqnarray*}
\end{proof}
Till now, we  find the $t^*_{l,k}$ additional flows constitute a nice quantum torus algebra.  Next, similar to the KP hierarchy,
it is natural to consider the quantum torus constraint on the tau function of the BKP hierarchy.

In \cite{tu}, one has shown that
\begin{eqnarray}
\partial_{t_{p,s}}\log w_B=(e^{\tilde\eta}-1)\frac{\frac{Z_s^{(p+1)}}{p+1}(\tau_B)
}{\tau_B},
\end{eqnarray}
where $Z_s^{(p+1)}$ is the generator of $W^B_{\infty}$ algebra.
Then with the help of rewriting  the quantum torus flow $\partial_{t^*_{l,k}}$ in terms of the $\partial_{t_{p,s}}$ flows
\begin{eqnarray*}
\partial_{t^*_{l,k}}&=&\sum_{p,s=0}^{\infty}\frac{l^p(k\log q)^s}{p!s!}\partial_{t_{p,s}},
\end{eqnarray*}
and denoting \[L^B_{l,k}:=\sum_{p,s=0}^{\infty}\frac{l^p(k\log q)^s}{p!s!}\frac{Z_s^{(p+1)}}{p+1},\]
 the quantum torus constraint on the wave function $w_B$, i.e.
\begin{eqnarray}
\partial_{t^*_{l,k}}w_B&=&0,
\end{eqnarray}
will lead to the quantum torus constraint on the  tau function of the BKP hierarchy
\[L^B_{l,k}\tau_B=c,\] \label{constraintontauBKP}
where $c$ is a constant.
Frow these above, we can find the remarkable difference of the quantum torus constraints on the tau function
of KP and BKP hierarchies, which originates from the B type condition in eq.(\ref{Bcondition}).

\section{Conclusions and Discussions}

In this paper, we construct the quantum torus symmetry of the  KP hierarchy and give the quantum torus flow equation on  the wave functions. Meanwhile the representation of the  quantum torus algebra over the tau function space was given. Further,  like Virasoro constraint, using ASvM formula we give a new constraint called the quantum torus constraint  on the tau function which might be useful in quantum field theory, supersymmetric gauge theory and so on. After that, by reduction we also construct the quantum torus symmetry of the KdV and  BKP hierarchies and further derived the quantum torus constraints on their tau functions. We are also looking forward to finding the application of the quantum torus symmetry of these reduced KP type integrable hierarchies.

{\bf Acknowledgments} {\noindent \small  Chuanzhong Li is supported by the National Natural Science Foundation of China under Grant No. 11201251, the Zhejiang Provincial Natural Science Foundation of China under Grant No. LY12A01007, the Natural Science Foundation of Ningbo under Grant No. 2013A610105. Jingsong He is supported by the National Natural Science Foundation of China under Grant No. 11271210, K. C. Wong Magna Fund in Ningbo University.  }



\vskip20pt
\begin{thebibliography}{AAA1}
\frenchspacing


\bibitem{dickey2} L. A. Dickey, Soliton Equations and Hamiltonian
Systems(2nd Edition) (World Scintific, Singapore,2003).



\bibitem{os}A. Yu. Orlov, E. I. Schulman, Additional symmetries of integrable equations and
conformal algebra reprensentaion, Lett. Math. Phys. 12(1986),
171-179.

\bibitem{heCKP} J. S. He, K. L. Tian, A. Forester, W. X. Ma,
Additional Symmetries and String Equation of the CKP Hierarchy, Lett. Math. Phys. 81(2007), 119-134.

\bibitem{tiancCKP}K.L. Tian, J. S. He, J. P. Cheng, Y. Cheng, Additional symmetries of constrained CKP and BKP hierarchies,
SCIENCE CHINA Mathematics 54(2011), 257-268.


\bibitem{maohuajmp}M. H. Li, C. Z. Li etal, Virasoro type algebraic structure hidden in the constrained discrete
Kadomtsev-Petviashvili hierarchy, J. Math. Phys. 54(2013), 043512.

  \bibitem{Block}
  R. Block, On torsion-free abelian groups and Lie algebras,  Proc. Amer. Math. Soc., 9(1958), 613-620.

  \bibitem{DZ}
D. Dokovic, K. Zhao, Derivations, isomorphisms and second cohomology
of generalized Block algebras, Algebra Colloq., 3(1996), 245-272.


  \bibitem{Su}
Y. Su, Quasifinite representations of a Lie algebra of Block type,
J. Algebra, 276(2004), 117-128.

  \bibitem{zwz}
M. R. Chen, et al., $W_{1+\infty}$ 3-algebra and KP Hierarchy, arXiv:1309.4627.

\bibitem{Qvirasoro}R. Kemmoku and S. Saito, Discretization of Virasoro
algebra,  Physics Letters B 319(1993), 471-477.


\bibitem{qwkp}J. Mas and M. Seco,  The algebra of q-pseudodifferential symbols and the $q-W_{KP}^{(n)}$ algebra, J. Math. Phys. 37(1996), 6510.



\bibitem{ourBlock}
 C. Z. Li, J. S. He, Y. C. Su, Block type symmetry of bigraded Toda hierarchy,
J. Math. Phys. 53(2012), 013517.







\bibitem{ourJMP}
 C. Z. Li, J. S. He, K. Wu, Y. Cheng,  Tau function and  Hirota bilinear equations for the extended  bigraded Toda
 Hierarchy, J. Math. Phys., 51(2010),043514.

 \bibitem{solutionBTH} C. Z. Li, Solutions of  bigraded Toda hierarchy,
Journal of Physics A 44(2011), 255201.


 \bibitem{dispBTH}
 C. Z. Li, J. S. He, Dispersionless bigraded Toda hierarchy and its additional symmetry, Rev.  Math. Phys., 24(2012), 1230003.






\bibitem{dtyptds}
 C. Z. Li, J. S. He,  Block algebra in two-component BKP and D type Drinfeld-Sokolov hierarchies, arXiv:1210.6498, J. Math. Phys. 54(2013), 113501.





\bibitem{sinealgebra}D. B. Fairlie, P. Fletcher,  C. K. Zachos, Trigonometric structure constants for new infinite-dimensional
algebras, Phys. Lett. B 218(1989), 203.





\bibitem{siebergwitten}T. Maeda, T. Nakatsu, K. Takasaki,  T. Tamakoshi, Free fermion and Seiberg-Witten differential in
random plane partitions. Nucl. Phys. B 715(2005), 275.

\bibitem{instanton}T. Maeda, T. Nakatsu, Amoebas and instantons. Internat. J. Modern Phys. A 22(2007), 937.

\bibitem{takasakiquantum}T. Nakatsu, K. Takasaki, Melting Crystal, Quantum Torus and Toda Hierarchy, Commun. Math. Phys. 285(2009), 445-468.





\bibitem{dickeyacta} L. A. Dickey, Lectures on classical W-algebras, Acta Appl. Math. 47(1997), 243-321.
\bibitem{witten}E. Witten, Two-dimensional gravity and intersection theory on moduli space, Surveys
in Diff. Geom. 1(1991), 243-310.
\bibitem{kontsevich}M. Kontsevich, Intersection theory on the moduli space of curves and the matrix
Airy function, Commun. Math. Phys. 147(1992), 1-23.
\bibitem{Douglas} M. Douglas, Strings in less than one dimension and the generalized
KdV hierarchies, Phys. Lett. B 238(1990),176-180.


\bibitem{DKJM-KPBKP}
E. Date, M. Kashiwara, M. Jimbo, T. Miwa, Transformation groups for
soliton equations. Nonlinear integrable systems--classical theory
and quantum theory (Kyoto, 1981), 39-119, World Sci. Publishing,
Singapore, 1983.

\bibitem{kt1}K. Takasaki, Quasi-classical limit of BKP hierarchy and $W$-infnity symmetris,
Lett. Math. Phys. 28(1993), 177-185.





\bibitem{tu}M. H. Tu, On the BKP Hierarchy: Additional Symmetries,
Fay Identity and Adler-Shiota-van Moerbeke
Formula, Lett. Math. Phys.  81(2007),93-105.


\bibitem{kodamaWalgebra}
  S. Aoyama, Y. Kodama, A generalized Sato equation and the $W_{\infty}$ algebra,
  Phys. Lett. B, 278(1992), 56-62.

\bibitem{asv94}M. Adler, T. Shiota, P. van Moerbeke, From the $w_{\infty}$-algebra to its central extension:
a $\tau$-function approach. Phys. Lett. A 194(1994), 33-43.




\bibitem{asv2}M. Adler, T. Shiota, P. van Moerbeke, A Lax representation for the Vertex operator
and the central extension, Comm. Math. Phys. 171(1995), 547-588.












\end{thebibliography}
\end{document}